\documentclass[12pt, letterpaper]{article}
\usepackage[margin=1in]{geometry}
\usepackage{authblk}

\usepackage{apacite}
\usepackage{natbib}
\usepackage{url}
\usepackage{amsthm}
\newtheorem{theorem}{Theorem}

\usepackage{amsmath}
\usepackage{graphicx}
\usepackage{subcaption}
\graphicspath{{./figures/}}
\DeclareGraphicsExtensions{.png, .pdf}
\usepackage{tikz}
\usetikzlibrary{arrows}
\usepackage[markup=underlined]{changes}

\begin{document}

\title{\textbf{A General Approach to Seismic Inversion with Automatic Differentiation}}

\author[1*]{Weiqiang Zhu}
\author[2*]{Kailai Xu}
\author[2,3]{Eric Darve}
\author[1]{Gregory C. Beroza}
\affil[1]{\small Department of Geophysics, Stanford University, Stanford, CA, 94305}
\affil[2]{\small Institute for Computational and Mathematical Engineering, Stanford University, Stanford, CA, 94305}
\affil[3]{\small Mechanical Engineering, Stanford University, Stanford, CA, 94305}
\affil[ ]{\texttt {\{zhuwq, kailaix, darve, beroza\}@stanford.edu}}
\date{}
\renewcommand{\thefootnote}{\fnsymbol{footnote}}
\footnotetext[1]{The two authors contributed  equally to this paper.}

\maketitle

\section*{Abstract}
Imaging Earth structure or seismic sources from seismic data involves minimizing a target misfit function, and is commonly solved through gradient-based optimization. The adjoint-state method has been developed to compute the gradient efficiently; however, its implementation can be time-consuming and difficult. We develop a general seismic inversion framework to calculate gradients using reverse-mode automatic differentiation. The central idea is that adjoint-state methods and reverse-mode automatic differentiation are mathematically equivalent. The mapping between numerical PDE simulation and deep learning allows us to build a seismic inverse modeling library, ADSeismic, based on deep learning frameworks, which supports high performance reverse-mode automatic differentiation on CPUs and GPUs. We demonstrate the performance of ADSeismic on inverse problems related to velocity model estimation, rupture imaging, earthquake location, and source time function retrieval. ADSeismic has the potential to solve a wide variety of inverse modeling applications within a unified framework.  

\section*{Introduction}

Inverse modeling is used in seismology to recover physical parameters such as earthquake location, magnitude, and Earth's interior structure. Such inverse problems are usually solved by minimizing a misfit function that measures the discrepancy between predictions and observations. Gradient-based optimization requires calculation of the gradient of the misfit function with respect to the physical parameters. The adjoint-state method \citep{plessix2006review} is a commonly used technique for computing the gradient efficiently. This method solves an adjoint linear system, which involves solutions of the forward problem. The drawback of the adjoint-state method is that the derivation and implementation can be very challenging, and must be done on a case-by-case basis for different systems.
To our knowledge, although many frameworks exist for specific inverse modeling applications  \citep{rucker2017pygimli, cockett2015simpeg}, general frameworks that can estimate physical parameters without case-by-case gradient derivation and implementation are lacking. 

Automatic differentiation (AD) \citep{paszke2017automatic, baydin2017automatic}, where the gradients are computed automatically based on the computational graph of the forward simulation, provides an alternative approach. In AD, a computational graph of the forward simulation keeps track of arithmetical operation dependencies, stores intermediate results, and computes the gradient using the chain rule. AD has been the dominant approach for training deep neural networks, which is known as ``backpropagation'' in the deep learning community. Both deep neural networks and PDE simulations can be viewed as a series of linear or nonlinear operators \citep{hughes2019wave}. Moreover, reverse-mode automatic differentiation has been shown to be equivalent to the adjoint-state method mathematically \citep{li2019time}. This correspondence allows us to develop a flexible and general seismic inversion framework, ADSeismic, based on current deep learning frameworks such as TensorFlow \citep{abadi2016tensorflow} and PyTorch \citep{paszke2019pytorch}.  ADseismic provides a high performance environment with easily accessible gradients on CPUs, GPUs, and TPUs~\citep{jouppi2017datacenter}. 

We note that AD has already been applied to velocity estimationg in exploration seismology~\citep{sambridge2007automatic, cao2015computational, vlasenko2016efficiency,  richardson2018seismic}. In contrast to existing open-sourced seismic inversion software; however, ADSeismic is built on a deep learning framework, which allows for flexibly experimenting with new models, leverages specialized hardware designed for deep learning, and executes numerical simulations on heterogeneous computing platforms.

 We demonstrate several applications, including: velocity estimation, fault rupture imaging, earthquake location, and source time function retrieval. AD yields the same results as adjoint-state methods. The advantage is that while we need to derive and implement a specific gradient in each case with adjoint-state methods, different inversion problems can be solved with little or no change in the forward simulation codes with ADSeismic. Moreover, we achieve more than 20 times and 60 times acceleration for acoustic and elastic wave equations respectively when switching to GPU devices compared to CPUs. Since deep learning hardware and frameworks are improving continuously, ADSeismic provides seismic inverse modeling with increasingly  powerful automatic differentiation techniques for a wide range of applications.

\section*{Method}

\subsection*{Automatic Differentiation}
Automatic differentiation (AD) is a general and efficient method to compute gradients based on the chain rule. By tracing the forward-pass computation, the gradient at the final step propagates back to each operator and parameter in a computational graph. AD is mainly used for training neural network models that consist of a sequence of linear transforms and non-linear activation functions. AD calculates the gradients of every variable by propagating the gradients back from the loss function to the trainable parameters. These gradients are then used in a gradient-based optimizer, such as the gradient descent (GD) method to update the parameters and minimize the differences between the model predictions and the ground-truth labels. Numerical simulations based on PDEs are similar to neural network models in that they are both sequences of linear/non-linear transformations (Fig. \ref{fig:compare-NN-PDE}). For example, the Finite-Difference Time-Domain (FDTD) method~\citep{yee1966numerical}, applies a finite difference operator to consecutive time steps to solve time-dependent PDEs~\citep{hughes2019wave}. In seismic problems, we specify parameters, such as wave velocity, source location, or source time functions, in forward simulations to generate predicted seismic signals. 
In ADSeismic, the gradients of the observational differences over these parameters can be computed automatically and thus used in a gradient-based optimizer in the same way as when training neural networks.  

\begin{figure}[htpb]
    \centering
    \includegraphics[width=0.8\linewidth]{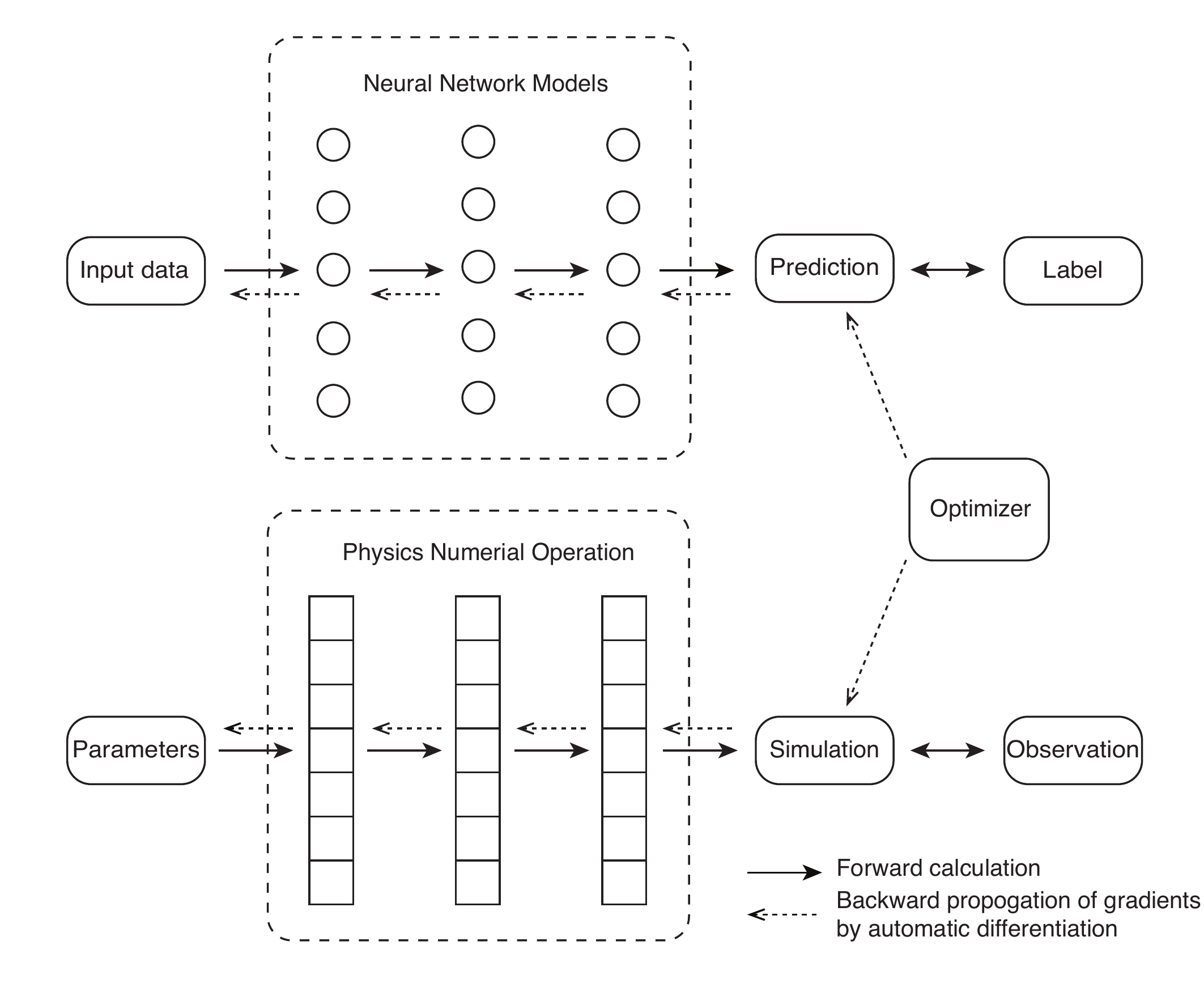}
    \caption{The similarity between neural networks and PDE-based physical simulation}
    \label{fig:compare-NN-PDE}
\end{figure}

\subsection*{Relationship to the Adjoint Method}

The adjoint-state method is an efficient technique for computing the gradient of the misfit function with respect to the physical parameters of interest. For example, the adjoint-state method is commonly used to compute the gradient in full-waveform inversion \citep{plessix2006review}. To clarify the connection between the adjoint-state method and reverse-mode automatic differentiation, we provide a derivation based on the Lagrange multipliers.

Consider the explicit discretization of the wave equation, which can be written as 
\begin{equation}
\begin{aligned}
    U_1 &= A(\theta) U_0 + F_0 \\
    U_2 &= A(\theta) U_1 + F_1\\
    &\vdots \\
    U_{n-1} &= A(\theta) U_{n-2} + F_{n-2}\\
    U_{n} &= A(\theta) U_{n-1} + F_{n-1}\\
\end{aligned}
\label{equ:forward}
\end{equation}
where $U_{k}$ is the seismic wavefield at $k$-th time step, $F_{k}$ is the source term at the $k$-th time step, $A(\theta)$ is the associated coefficient matrix, and $\theta$ is the physical parameter of interest, e.g., the wave velocity. $A(\theta)$ indicates that the entries in the matrix depend on $\theta$. 
To simplify the notation, we let the misfit function be
$$J(\theta) = \frac{1}{2} \sum_{k=1}^n \|U_k(\theta)-U_k^{\mathrm{obs}}\|^2$$
where $U_k^{\mathrm{obs}}$ is the observation at $k$-th step. 
The corresponding Lagrangian functional is 
\begin{equation}\label{equ:Lag}
    L(\theta, U_{1}, \ldots, U_{n}) = \frac{1}{2} \sum_{k=1}^{n} ||U_k-U_k^{\mathrm{obs}}||^2 + \sum_{k=1}^{n} {\lambda_k}^T \left( A(\theta) U_{k-1} + F_{k-1} - U_k \right)
\end{equation}
where $\lambda_k$ is the adjoint variable. 
The Karush–Kuhn–Tucker (KKT) condition \citep{luenberger1984linear} for Eq.~\eqref{equ:Lag} reads
\begin{equation}\label{equ:Adj1}
\begin{aligned}
    \frac{\partial L}{\partial U_{n}} &= U_{n}  - U_{n}^{\mathrm{obs}} - \lambda_{n} = 0 \\
    \frac{\partial L}{\partial U_{n-1}} &= U_{n-1} - U_{n-1}^{\mathrm{obs}}- \lambda_{n-1} + A(\theta)^T \lambda_{n} = 0 \\
    &\vdots \\
    \frac{\partial L}{\partial U_{2}} &= U_{2} - U_{2}^{\mathrm{obs}}- \lambda_{2} + A(\theta)^T \lambda_{3} = 0 \\
    \frac{\partial L}{\partial U_{1}} &= U_{1} - U_{1}^{\mathrm{obs}}- \lambda_{1} + A(\theta)^T \lambda_{2} = 0 \\
\end{aligned}
\end{equation}

Rearranging \eqref{equ:Adj1} we obtain
\begin{equation}
\begin{aligned}
    \lambda_{n} &= U_{n}- U_{n}^{\mathrm{obs}}\\
    \lambda_{n-1} &= A(\theta)^T\lambda_{n} + U_{n-1} - U_{n-1}^{\mathrm{obs}} \\
    &\vdots \\
    \lambda_{2}  &= A(\theta)^T \lambda_{3} + U_{2} - U_{2}^{\mathrm{obs}} \\
    \lambda_{1} &= A(\theta)^T \lambda_{2} + U_{1} - U_{1}^{\mathrm{obs}} \\
\end{aligned}
\label{eqn:adjoint-backward}
\end{equation}
Note that we can compute all the adjoint variables $\lambda_{k}, k=1,2,\ldots, n$ sequentially from $k=n$ to $k=1$. In this process, we need to perform matrix multiplication with the coefficient matrix $A(\theta)^T$, which why we call $\lambda_{k}$ adjoint variables. 

Finally, the gradients of $L$ with respect to $\theta$ can be extracted using the computed $\lambda_{k}, k=1,2,\ldots, n$
\begin{equation}
    \boxed{\frac{\partial L}{\partial \theta} = \sum_{k=1}^{n} {\lambda_k}^T \frac{\partial A(\theta)}{\partial \theta} U_{k-1} }
    \label{eqn:adjoint-gradient}
\end{equation}

In the following text, we describe how reverse-mode AD is used for computing the gradient $\frac{\partial J}{\partial \theta}$ and show that AD calculates the adjoint variables and gradients in the same way as the adjoint-state method (Eq. \eqref{eqn:adjoint-backward} and \eqref{eqn:adjoint-gradient}). A straightforward way to view AD is to consider a specific operator in the computational graph from $k-1$ to $k$ step: 
\begin{equation}
\begin{aligned}
    \text{Forward Computation:  }\;& U_k(U_{k-1}, \theta) = A(\theta) U_{k-1}+F_{k-1}\\
    \text{Backward Gradient:  }\;&
    \frac{\partial U_k(U_{k-1}, \theta)}{\partial U_{k-1}} = A(\theta)^T \\
    \;& \frac{\partial U_k(U_{k-1}, \theta)}{\partial \theta} = \frac{\partial A(\theta)}{\partial \theta} U_{k-1} \\
\end{aligned}
\end{equation}

We assume that the gradient of $J$ with respect to $U_{k}$ has already been calculated at the $k$-th time step. We then back-propagate the gradients to the previous time step (Fig.~\ref{fig:auto-diff}). For convenience we define
\begin{equation}\label{equ:muk}
\begin{aligned}
\mu_{k} &:= \left(\frac{\partial J(\theta, U_1, \ldots, U_k)}{\partial U_{k}} \right)^T& k = 1,2, \ldots, n-1\\ 
\mu_n &= \left(U_{n}-U_{n}^{\mathrm{obs}}\right)^T
\end{aligned}
\end{equation}
Here, $J(\theta, U_1, \ldots, U_k)$ can be recursively defined as 
\begin{align*}
J(\theta, U_1, \ldots, U_n) &:= \frac{1}{2}\sum_{k=1}^n \|U_k - U_k^{\mathrm{obs}}\|^2 \\
J(\theta, U_1, \ldots, U_k) &:= J(\theta, U_1, \ldots, U_{k}, A(\theta) U_{k}+F_{k}) & k = 1,2, \ldots, n-1
\end{align*}
where we define $J(\theta, U_1, \ldots, U_k)$ by substituting $U_{k+1}$ in $J(\theta, U_1, \ldots, U_k, U_{k+1})$ with $A(\theta) U_{k}+F_{k}$.

We now focus on one specific step shown in bold in Fig.~\ref{fig:auto-diff}. In AD, we need to compute the gradients $\frac{\partial J}{\partial U_{k-1}}$ and $\frac{\partial J}{\partial \theta}$ given the so-called ``top'' gradients $\frac{\partial J}{\partial U_{k}}$ (noted by the symbol ``b'' in Fig.~\ref{fig:auto-diff}).
The gradient backpropagation rule for $\frac{\partial J}{\partial U_{k-1}}$ reads
\begin{equation}
\begin{aligned}
\mu_{k-1}^T = \frac{\partial J(\theta, U_{1}, \ldots, U_{k-1})}{\partial U_{k-1}} &= \underbrace{\overbrace{\frac{\partial J(\theta, U_1,\ldots, U_{k})}{\partial U_{k}}}^{(b)}\frac{\partial U_k(U_{k-1}, \theta)}{\partial U_{k-1}}}_{(a)} + \overbrace{\frac{\partial J(\theta,U_1, \ldots, U_{n})}{\partial U_{k-1}}}^{(c)}\\
&=  {A(\theta)^T} \mu_k + \left(U_{k-1} -  U_{k-1}^{\mathrm{obs}}\right) \quad k = 2,\ldots, n\\
\end{aligned}
\label{eqn:ad-backward}
\end{equation}
Note $J$ on the left hand side and on the right hand side have different arguments. See 
Fig.~\ref{fig:auto-diff} for illustration. 

The gradient back-propagation rule for $\frac{\partial J(\theta, U_1, U_{2},\ldots, U_{k-1})}{\partial \theta}$ reads
\begin{equation*}
\begin{aligned}
g_{k-1} := \overbrace{\frac{\partial J(\theta, U_1, \ldots, U_{k-1})}{\partial \theta}}^{(d)} = \overbrace{\frac{\partial J(\theta, U_1, \ldots, U_{k})}{\partial U_{k}}}^{(b)} \frac{\partial U_{k}(U_{k-1},\theta)}{\partial \theta}
= \mu_k^T \frac{\partial A(\theta)}{\partial \theta} U_{k-1}
\end{aligned}
\end{equation*}

The gradient $\frac{\partial J(\theta)}{\partial \theta}$ is computed by accumulating $g_k$ from all steps
\begin{equation}\label{equ:ad-gradient}
   \boxed{ \frac{\partial J(\theta)}{\partial \theta} = \sum_{k=1}^n g_k = \sum_{k=1}^n \mu_k^T\frac{\partial A(\theta)}{\partial \theta}U_{k-1} }
\end{equation}

We now demonstrate the equivalence of the gradients (Eq.~\eqref{eqn:adjoint-gradient})  computed using AD and the gradients (Eq.~\eqref{equ:ad-gradient}) computed using the adjoint-state method.
\begin{theorem}\label{theorem:main}
    Assume that $\{\mu_k\}_{k=1}^n$ satisfies Eq.~\eqref{equ:muk} and Eq.~\eqref{eqn:ad-backward}, and $\{\lambda_k\}_{k=1}^n$ satisfies Eq.~\eqref{equ:Adj1}, then 
    \begin{equation}
        \lambda_k = \mu_k \quad k = 1,2, \ldots, n
    \end{equation}
    And therefore, 
    \begin{equation}\label{equ:thm}
        \sum_{k=1}^n \mu_k^T\frac{\partial A(\theta)}{\partial \theta}U_{k-1} = \sum_{k=1}^{n} {\lambda_k}^T \frac{\partial A(\theta)}{\partial \theta} U_{k-1}
    \end{equation}
\end{theorem}

\begin{proof}
Note $\lambda_{n} = \mu_{n} = U_{n} - U_{n}^{\mathrm{obs}}$ and the recursive relations Eq.~\eqref{eqn:adjoint-backward} and Eq.~\eqref{eqn:ad-backward} are the same. Thus, we have $\lambda_k = \mu_k, \quad k = 1,2, \ldots, n$.  Therefore, Eq.~\eqref{equ:thm} holds. 
\end{proof}
Theorem.~\ref{theorem:main} implies that reverse-mode automatic differentiation is mathematically equivalent to the adjoint-state method, and the intermediate gradient $\mu_k=\frac{\partial J}{\partial U_{k}}$ is exactly the adjoint variable $\lambda_k$. In the following text, we describe our general approach for seismic inversion based on the connection between the automatic differentiation and the adjoint state method.

\begin{figure*}[!ht]
	\centering
	\includegraphics[width=0.8\textwidth]{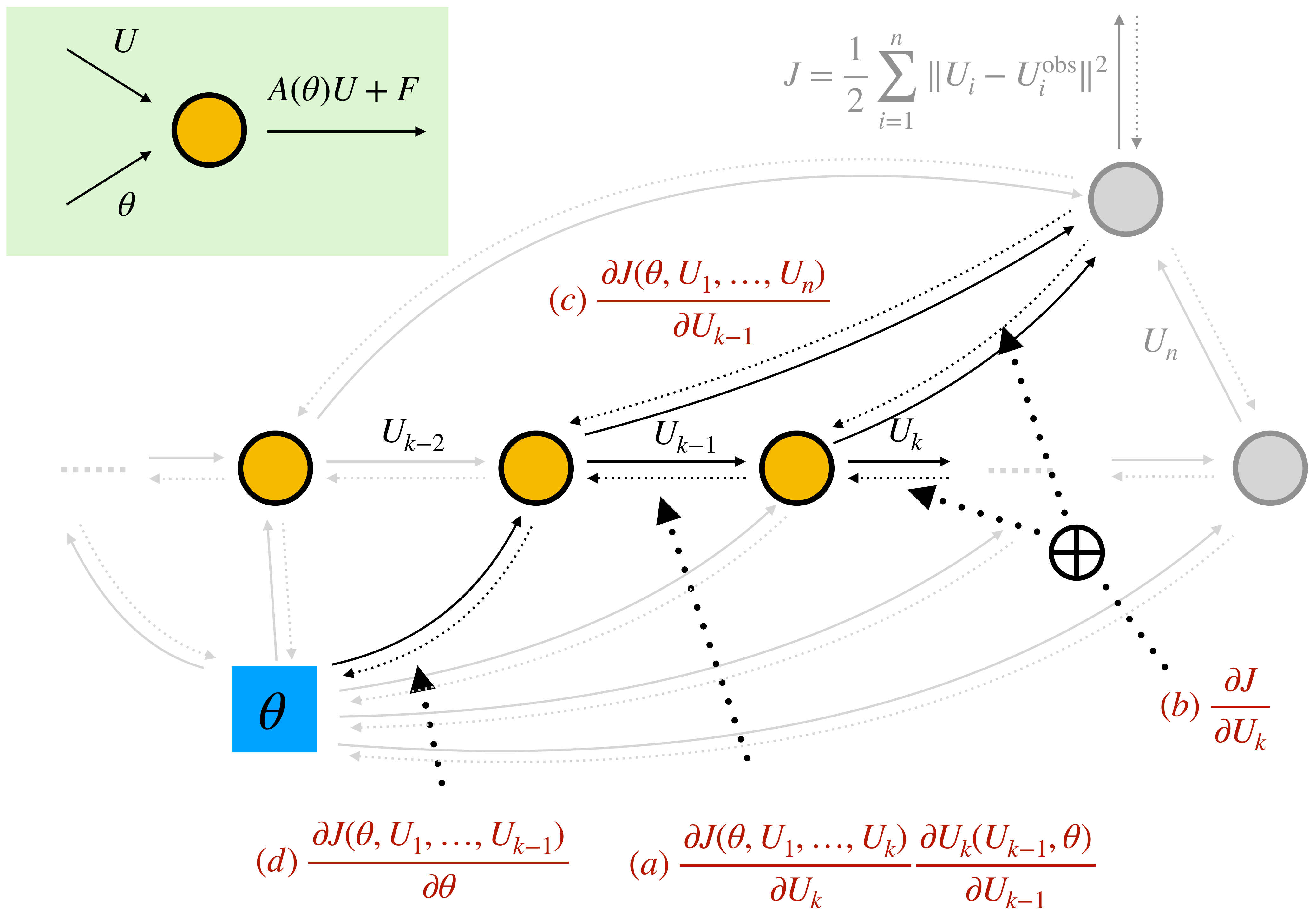}
	\caption{Computational graph and gradient back-propagation by automatic differentiation}
	\label{fig:auto-diff}
\end{figure*}

\subsection*{Implementation}

In this section we describe how automatic differentiation assists computing the gradient of the misfit function with respect to the physical parameters in ADSeismic. We use a staggered grid finite difference method for discretizing both the acoustic wave equation and the elastic wave equation with perfectly matched layer (PML)~\citep{roden2000convolution,komatitsch2007unsplit, grote2010efficient}. The governing equation for the acoustic wave equation is  
\begin{equation}
    \frac{\partial^2 u}{\partial t^2} = \nabla\cdot(c^2 \nabla u) +  f
    \label{eqn:acoustic}
\end{equation}
where $u$ is displacement, $f$ is the source term, and $c$ is the spatially varying acoustic velocity. The inversion parameters of interest are $c$ or $f$.
The governing equation for the elastic wave equation is 
\begin{equation}
\begin{aligned}
    \rho \frac{\partial v_i}{\partial t} &= \sigma_{ij, j} + \rho f_i \\ 
    \frac{\partial \sigma_{ij}}{\partial t} &= \lambda v_{k,k} + \mu(v_{i,j} + v_{j,i})
\end{aligned}
\label{eqn:elastic}
\end{equation}
where $v$ is velocity, $\sigma$ is stress tensor, $\rho$ is density, and $\lambda$ and $\mu$ are the Lam\'e's constants. The inversion parameters in the elastic wave equation case are $\lambda$, $\mu$, $\rho$ or $f$.

The finite difference discretization leads to a system of linear equations Eq.~\eqref{equ:forward} for both Eq.~\eqref{eqn:acoustic} and Eq.~\eqref{eqn:elastic}. For the adjoint-state method, we also need to derive and implement Eq.~\eqref{equ:Adj1} to compute the gradient Eq.~\eqref{eqn:adjoint-gradient}. This step is unnecessary in ADSeismic since the gradient is extracted automatically from the computational graph. We emphasize that only the forward simulation code is required for building a computational graph and the gradient automatically computed by AD is the same as that computed by the adjoint-state method. 

We use the Julia package, ADCME\footnote{\url{https://github.com/kailaix/ADCME.jl}}, for our implementation since it provides an interface to TensorFlow for automatic differentiation and intuitive Julia syntax for expressing mathematical formulae in numerical simulation. Additionally, ADCME provides built-in optimization solvers such as L-BFGS-B~\citep{zhu1997algorithm} for minimizing the misfit function. ADCME allows us to easily extend ADSeismic to other equations or models in seismic applications.

\section*{Applications}

In this section, we first highlight the performance of ADSeismic on CPUs and GPUs, where we observe an impressive 20- fold  and 60-fold acceleration for acoustic and elastic wave equations, respectively.   We then present three applications of ADSeismic to seismic problems including: velocity model estimation, earthquake location and source time function estimation, and  earthquake rupture imaging. The applications are built with the same forward simulation code (acoustic or elastic wave equations) with only minor changes to specify the inversion parameters to be recovered. 

\subsection*{Performance Benchmarking}

We benchmark the performance of ADSeismic\footnote{The CPU model on the test platform is the Intel(R) Xeon(R) CPU E5-2698 v4. The GPU model is the Tesla V100-SXM2.}. Since the backend of ADSeismic is TensorFlow, the same forward simulation code runs on both the CPU and GPU. The speed comparisons between the CPU and GPU for the acoustic equation and elastic equation are shown in Fig.~\ref{fig:speed_acoustic} and \ref{fig:speed_elastic} with the computation times averaged over three tests. We achieve more than 20 times the acceleration for the acoustic equation and 60 times the acceleration for the elastic equation on the GPU. The extra acceleration for the elastic equation is due to the fact that Tensorflow automatically parallels the updating of the velocity and stress tensors (Eq. \ref{eqn:elastic}). 

In ADSeismic, we can split the sources onto different GPUs so that the forward simulation and the associated gradient are computed using AD in parallel across the GPUs. Next, the gradients are assembled on the CPU and fed to the L-BFGS optimizer to update the inversion parameters (Fig.~\ref{fig:multi-gpu}). The updated inversion parameters are then distributed to all GPU devices for the next integration.
This multi-GPU routine avoids storing all wavefields on a single GPU, thus enabling us to perform larger numerical simulations than would otherwise be possible. 
\begin{figure*}[!ht]
	\centering
	\begin{subfigure}{0.45\textwidth}
		\includegraphics[width=\textwidth]{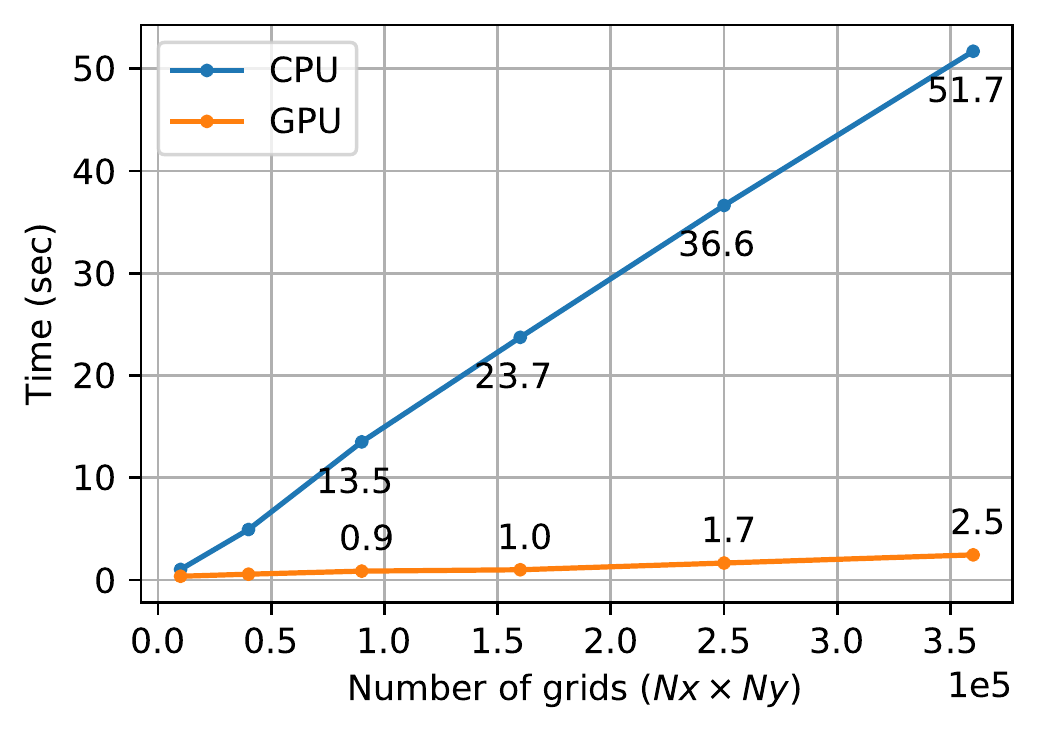}
		\caption{}
		\label{fig:speed_acoustic}
	\end{subfigure}
	\begin{subfigure}{0.45\textwidth}
		\includegraphics[width=\textwidth]{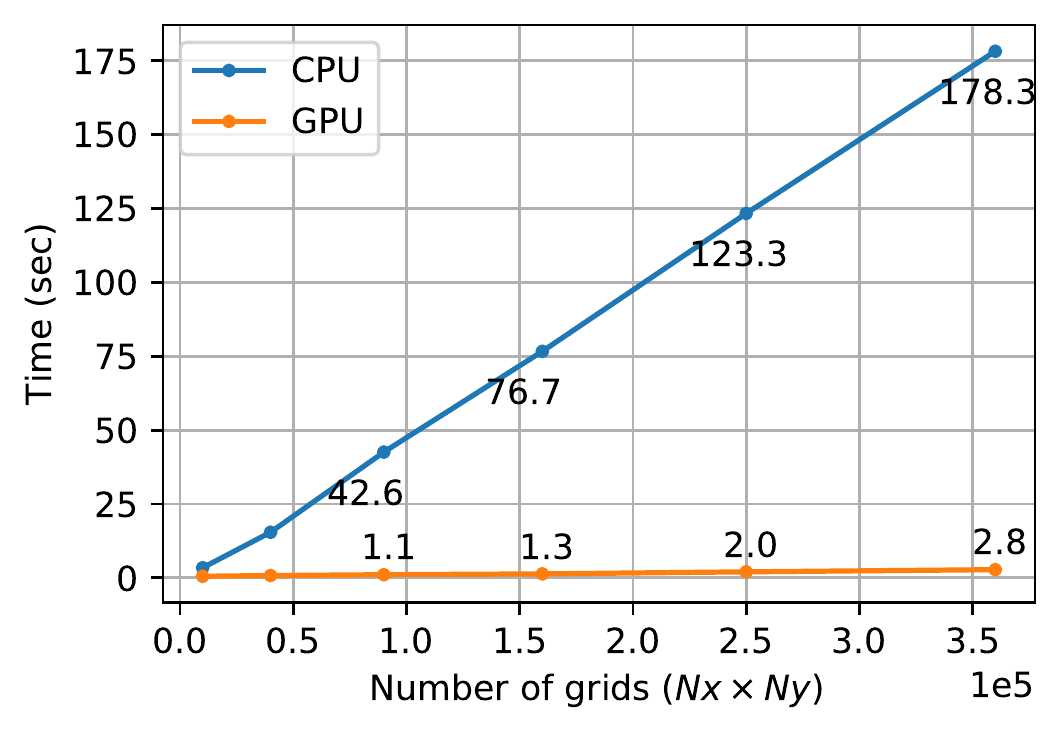}
		\caption{}
		\label{fig:speed_elastic}
	\end{subfigure}
	\begin{subfigure}{0.6\textwidth}
		\includegraphics[width=\textwidth]{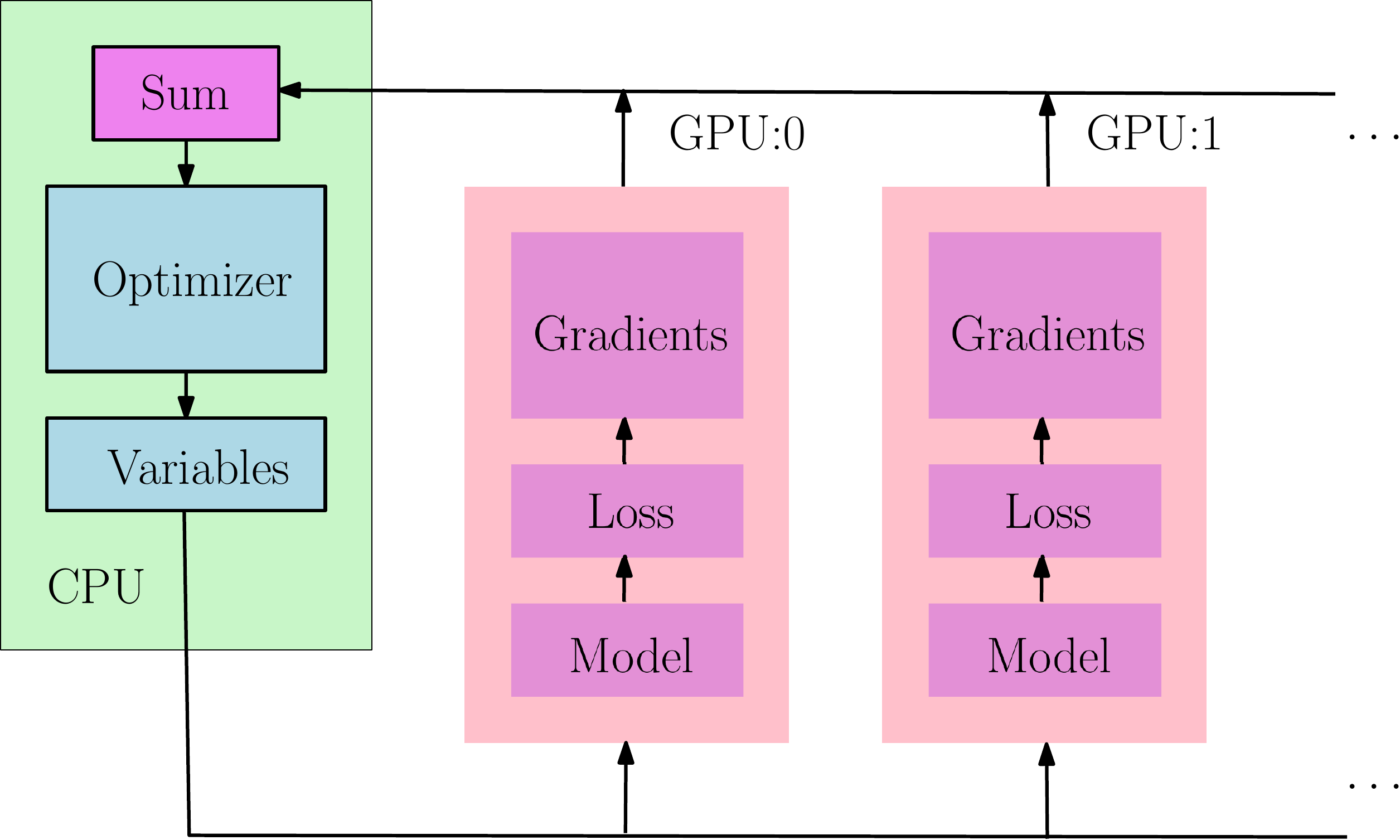}
		\caption{}
		\label{fig:multi-gpu}
	\end{subfigure}	
	\caption{High-performance computing of ADSeismic: (a) computational times on CPU and GPU for the acoustic equation; (b) computational times on CPU and GPU for the elastic equation; (c) diagram of multi-GPU computing.}
	\label{fig:speed}
\end{figure*}

\subsection*{Full-waveform Inversion}

Classic full-waveform inversion (FWI) is based on the adjoint-state method \citep{tarantola1984inversion, virieux2009overview, plessix2006review, fichtner2006adjoint}.
As shown above, AD is mathematically equivalent to the adjoint-state method so that we can apply AD directly to the full-waveform inversion without manual derivation of the adjoint-state equations.
We demonstrate our method using two cases: the well-known and geometrically complex Marmousi benchmark model \citep{versteeg1994marmousi, martin2002marmousi} (Fig.~\ref{fig:marmousi}) and a layered Earth crust model with embedded anomalies of elliptical shape (Fig.~\ref{fig:layermodel}).
We place eight active sources on the surface with a spacing of 850m for the Marmousi benchmark and four plane waves with incident angles from $-30^o$, $-10^o$, $10^o$, to$30^o$ from the bottom to mimic incoming teleseismic waves for the layered model. We use a Ricker wavelet as the source time function for both cases. 
Similar to common FWI applications, we choose the L-BFGS optimization method and a $L2-$norm loss function for all the inversion. 
We note that ADSeismic supports other optimization techniques such as the stochastic gradient descent (SGD) method~\citep{witte2018full, bottou2010large, richardson2018seismic} although the  application  and  comparison  of  these  optimizers is beyond the scope of this paper.
The inversion results in Fig.~\ref{fig:marmousi_inv} and Fig.~\ref{fig:layermodel_inv} show good recovery of the complex velocity structures and anomalies demonstrating that AD accurately estimates the velocity models. 

\begin{figure*}[htpb]
	\centering
	\begin{subfigure}{0.6\textwidth}
		\includegraphics[width=\textwidth]{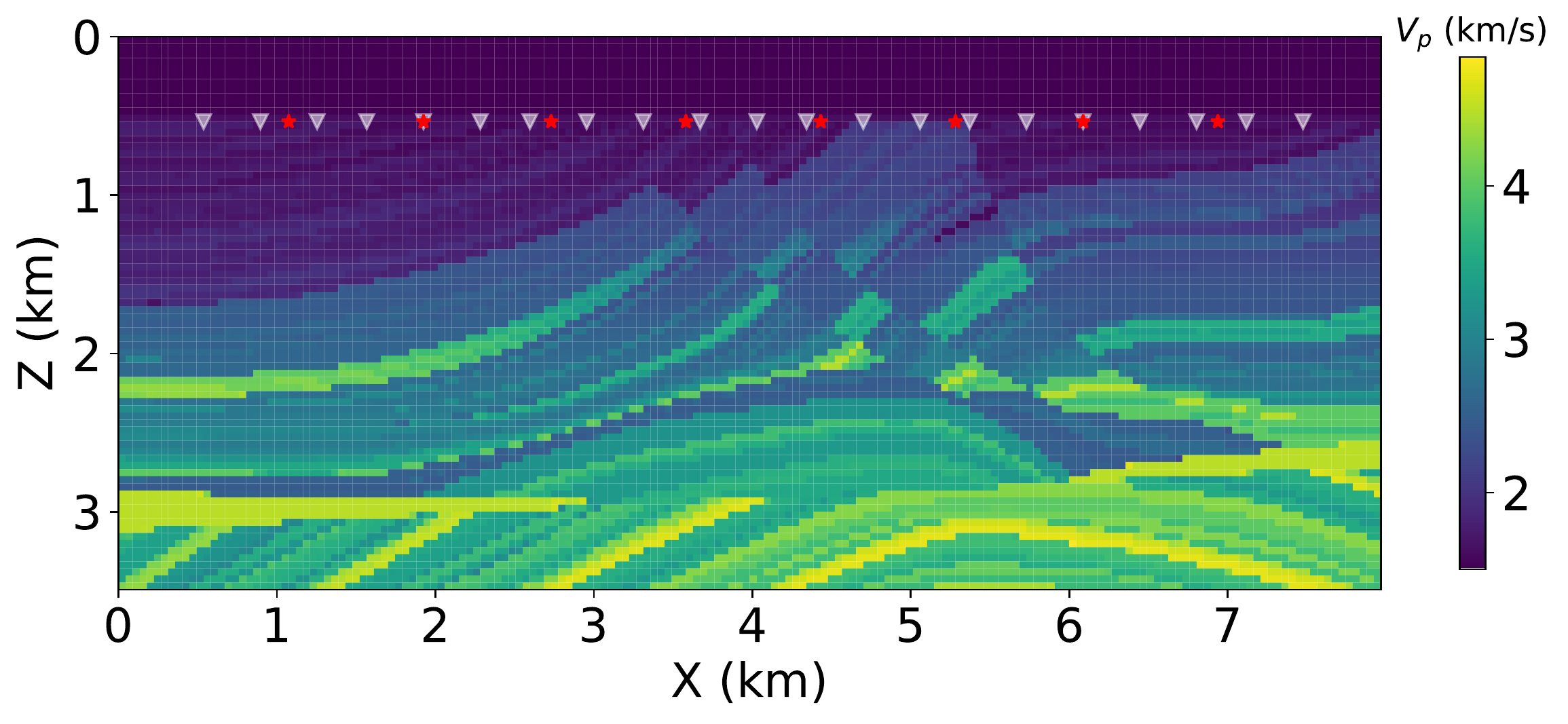}
		\caption{}
	\end{subfigure}
	\begin{subfigure}{0.6\textwidth}
		\includegraphics[width=\textwidth]{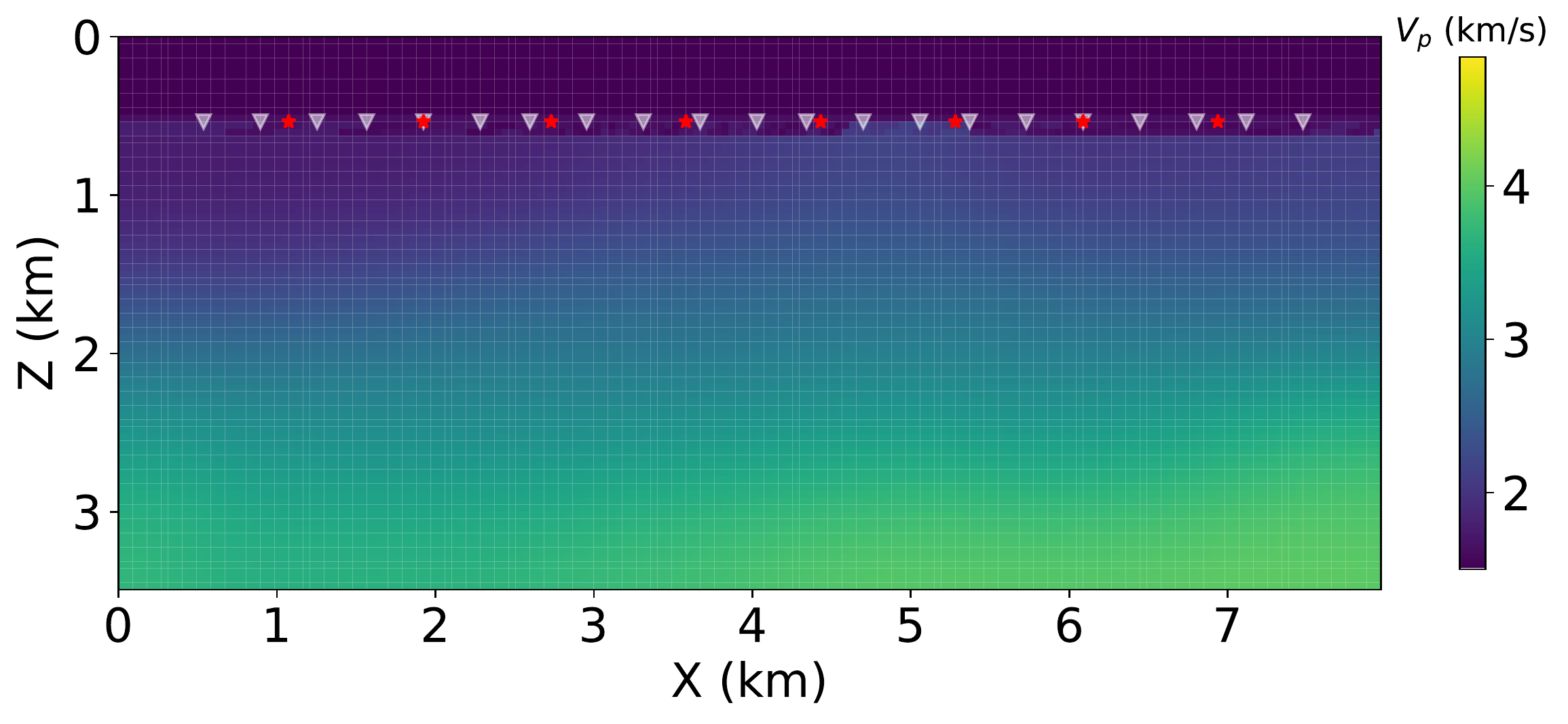}
		\caption{}
	\end{subfigure}	
	\begin{subfigure}{0.6\textwidth}
		\includegraphics[width=\textwidth]{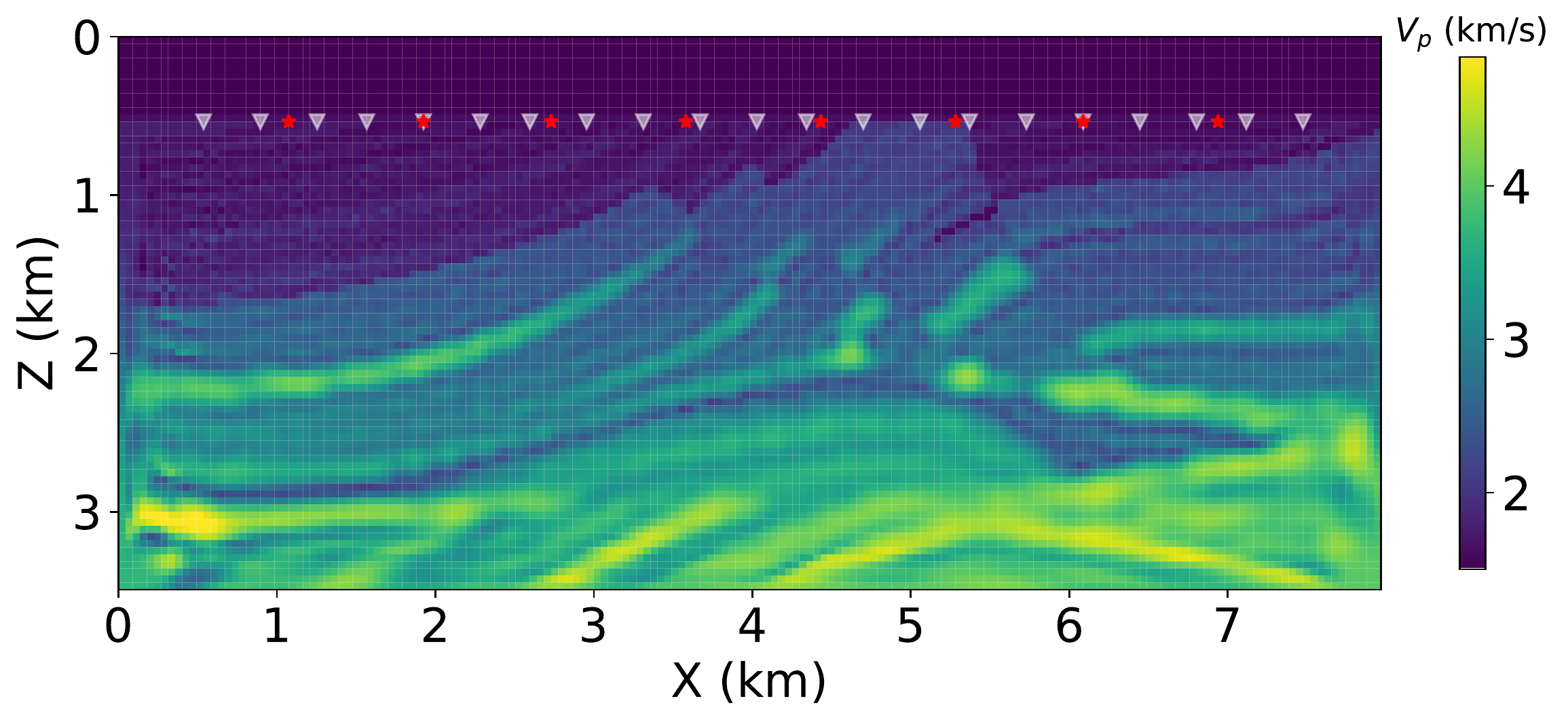}
		\caption{}
		\label{fig:marmousi_inv}
	\end{subfigure}	
	\caption{The Marmousi benchmark model: (a) the true P-wave velocity model; (b) the initial velocity model; (c) the inverted velocity model. The white triangles at the top represent the receiver locations, while the red stars represent the source locations.}
	\label{fig:marmousi}
\end{figure*}

\begin{figure*}[htpb]
	\centering
	\begin{subfigure}{0.6\textwidth}
		\includegraphics[width=\textwidth]{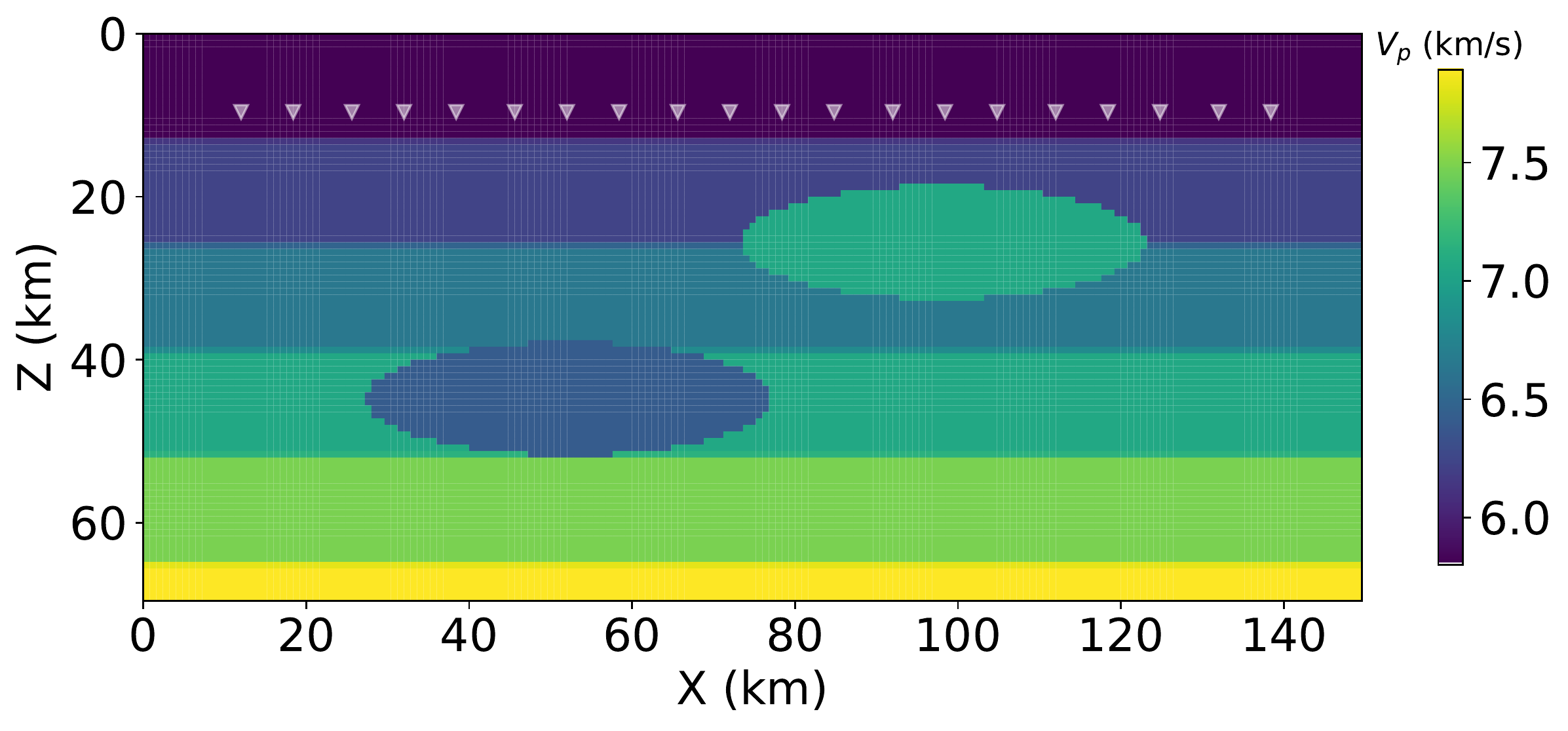}
		\caption{}
	\end{subfigure}
	\begin{subfigure}{0.6\textwidth}
		\includegraphics[width=\textwidth]{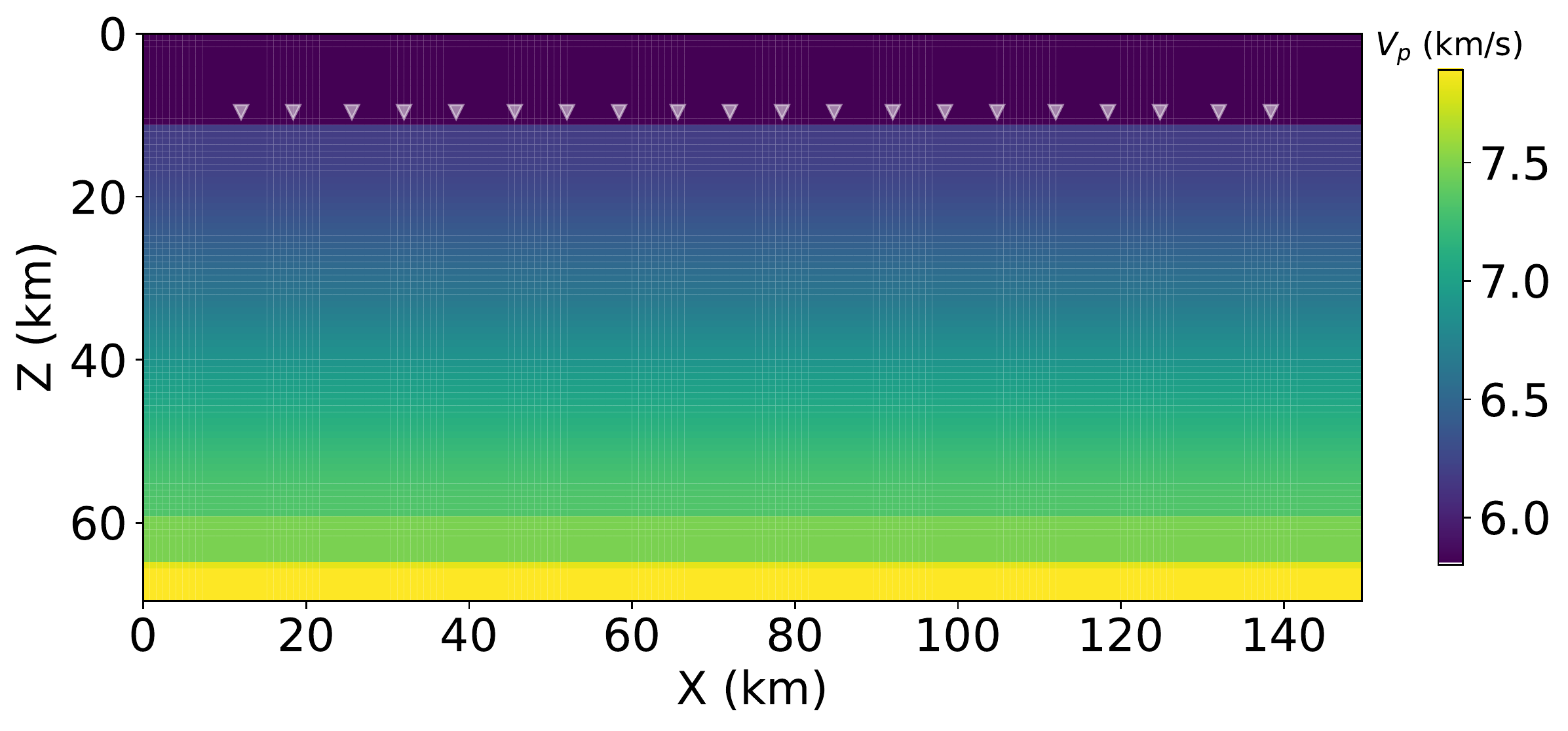}
		\caption{}
	\end{subfigure}	
	\begin{subfigure}{0.6\textwidth}
		\includegraphics[width=\textwidth]{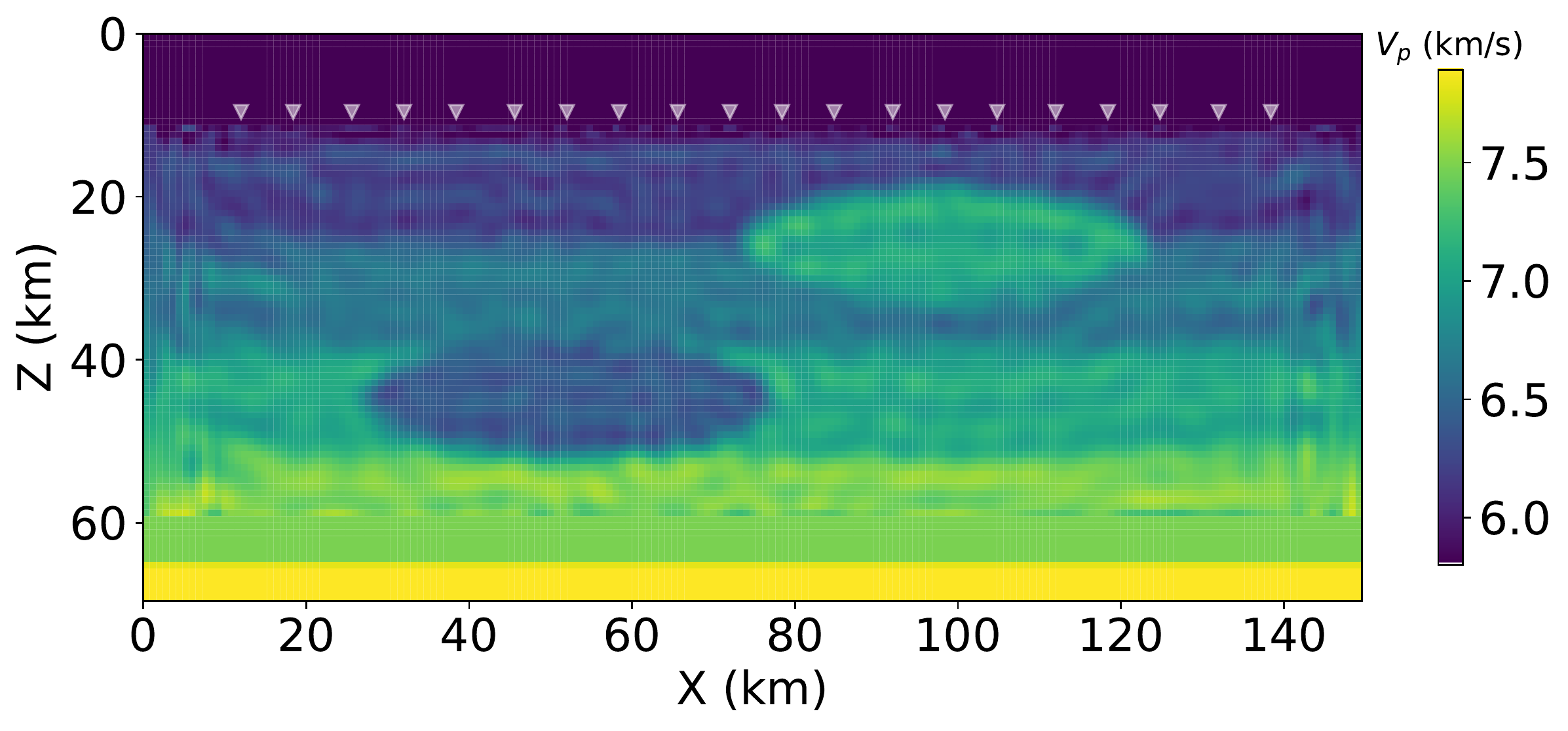}
		\caption{}
		\label{fig:layermodel_inv}
	\end{subfigure}	
	\caption{The layered model with inclusions: (a) the true P-wave velocity model; (b) the initial velocity model; (c) the inverted velocity model. Here we use four plane waves propagating from the bottom to the surface with incident angles of -30$^o$, -10$^o$, 10$^o$, and 30$^o$.}
	\label{fig:layermodel}
\end{figure*}

\subsection*{Earthquake Location and Source Time Function Retrieval}

Determining earthquake location is a routine, but essential earthquake monitoring task for which commonly used methods include 1) linearized inversion for absolute earthquake location \citep{lienert1986hypocenter, kissling1994initial, kissling1995program, klein2002user} and relative earthquake location \citep{waldhauser2000double, schaff2004optimizing}; 2) non-linear inversion methods \citep{thurber1985nonlinear, lomax2000probabilistic, lomax2009earthquake}; and 3) migration-based or time-reversal methods \citep{rubinstein2007full, Nakata2016, nakata2016migration}. The migration-based method produces a focused wavefield that is the same as the gradient in the first iteration of the adjoint-state method \citep{Fichtner2010}; however, this method does not explicitly give the source location but requires post-processing to extract potential earthquake locations from the focused wavefield. 

We use a new non-linear earthquake location method based on full waveforms.
The inversion target, the source term $f(x, t)$ in equation (\ref{eqn:acoustic}), is a delta function in space, whose gradient at zero is not well defined, making the direct application of the adjoint-state method difficult. 
With AD, we can flexibly re-parameterize the inversion target $f(x, t)$ with a continuous Gaussian form
\begin{equation}
    f(x, t) =  \frac{g(t)}{2\pi \sigma^2} \exp \left( -\frac{||x - x_0||^2}{2 \sigma^2} \right)
\end{equation}
where $g(t)$ is the source time function, $x_0$ is the earthquake location, and $\sigma$ is the standard deviation of the Gaussian function, which in our test is set to half of the grid size.
In this test, we simultaneously estimate the earthquake location $x_0$ and the source time function $g(t)$ by fitting the recorded waveforms. Fig.~\ref{fig:source} shows the evolution of the earthquake location and source time function during optimization from an initial state of a random selected earthquake location and a zero source time function.  The inversion results agree well with the true earthquake location and source time function. 

\begin{figure*}[htpb]
	\centering
	\begin{subfigure}{0.6\textwidth}
		\includegraphics[width=\textwidth]{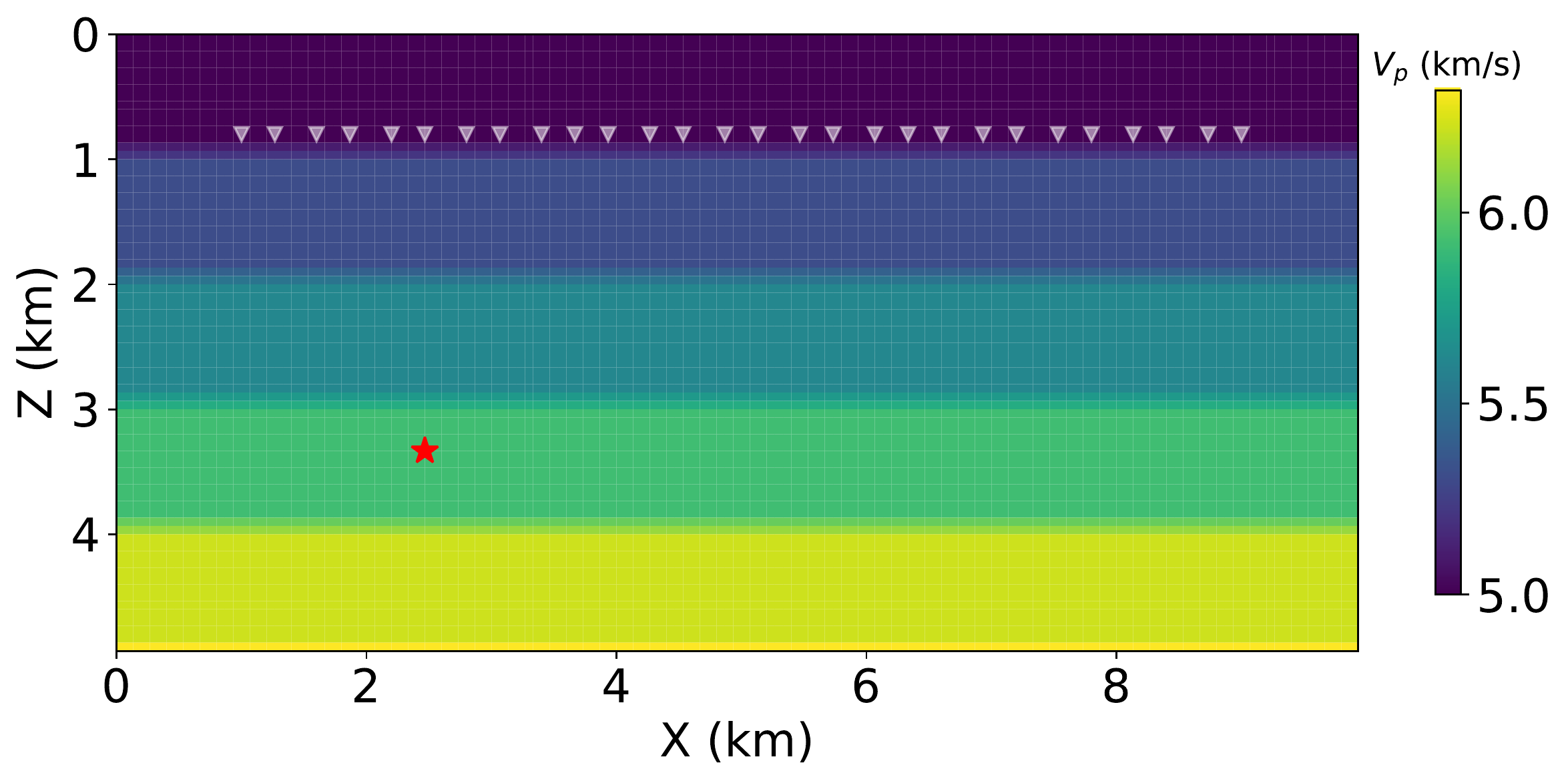}
		\caption{}
	\end{subfigure}
	\begin{subfigure}{0.6\textwidth}
		\includegraphics[width=\textwidth]{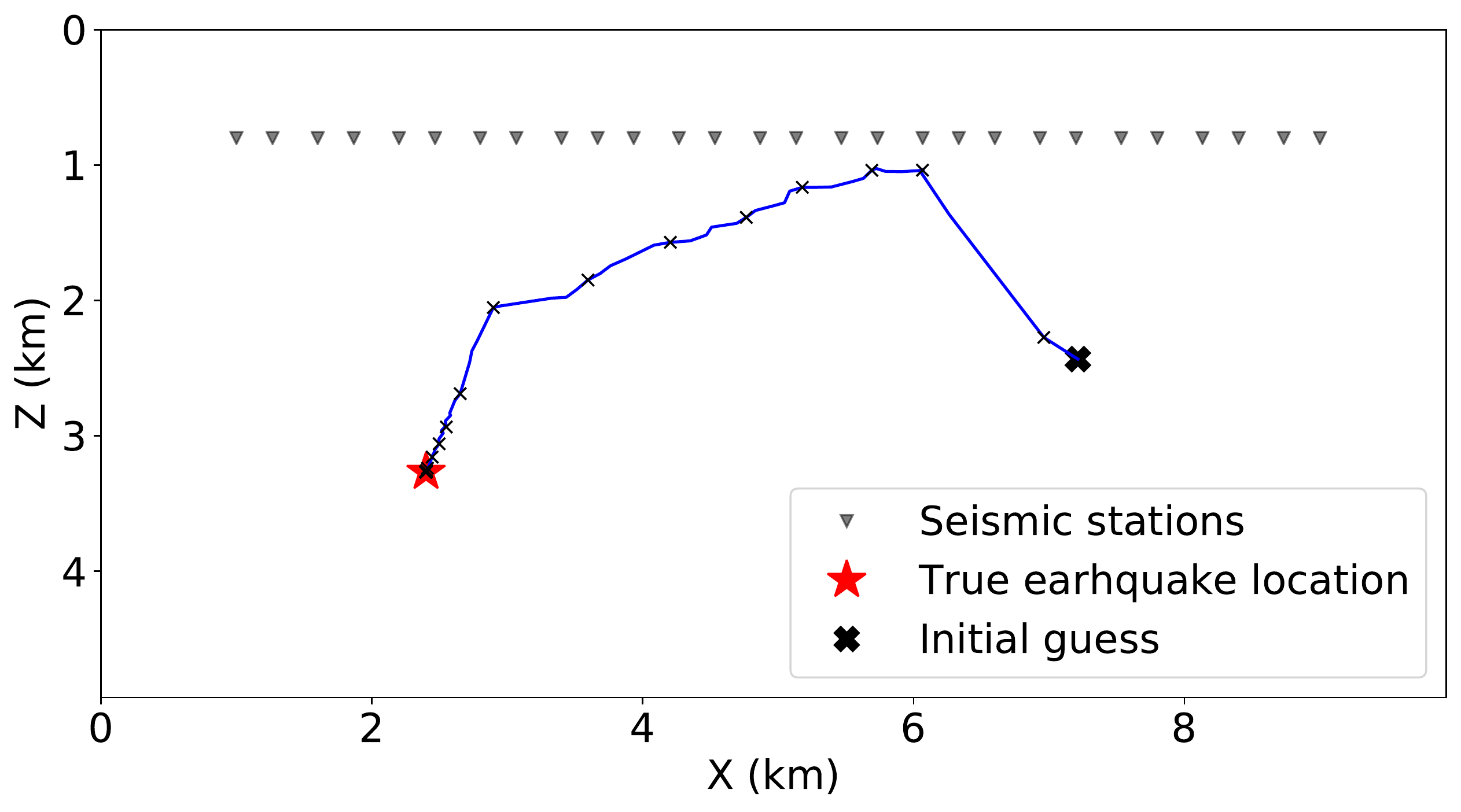}
		\caption{}
		\label{fig:source_loc}
	\end{subfigure}
	\begin{subfigure}{0.6\textwidth}
		\includegraphics[width=\textwidth]{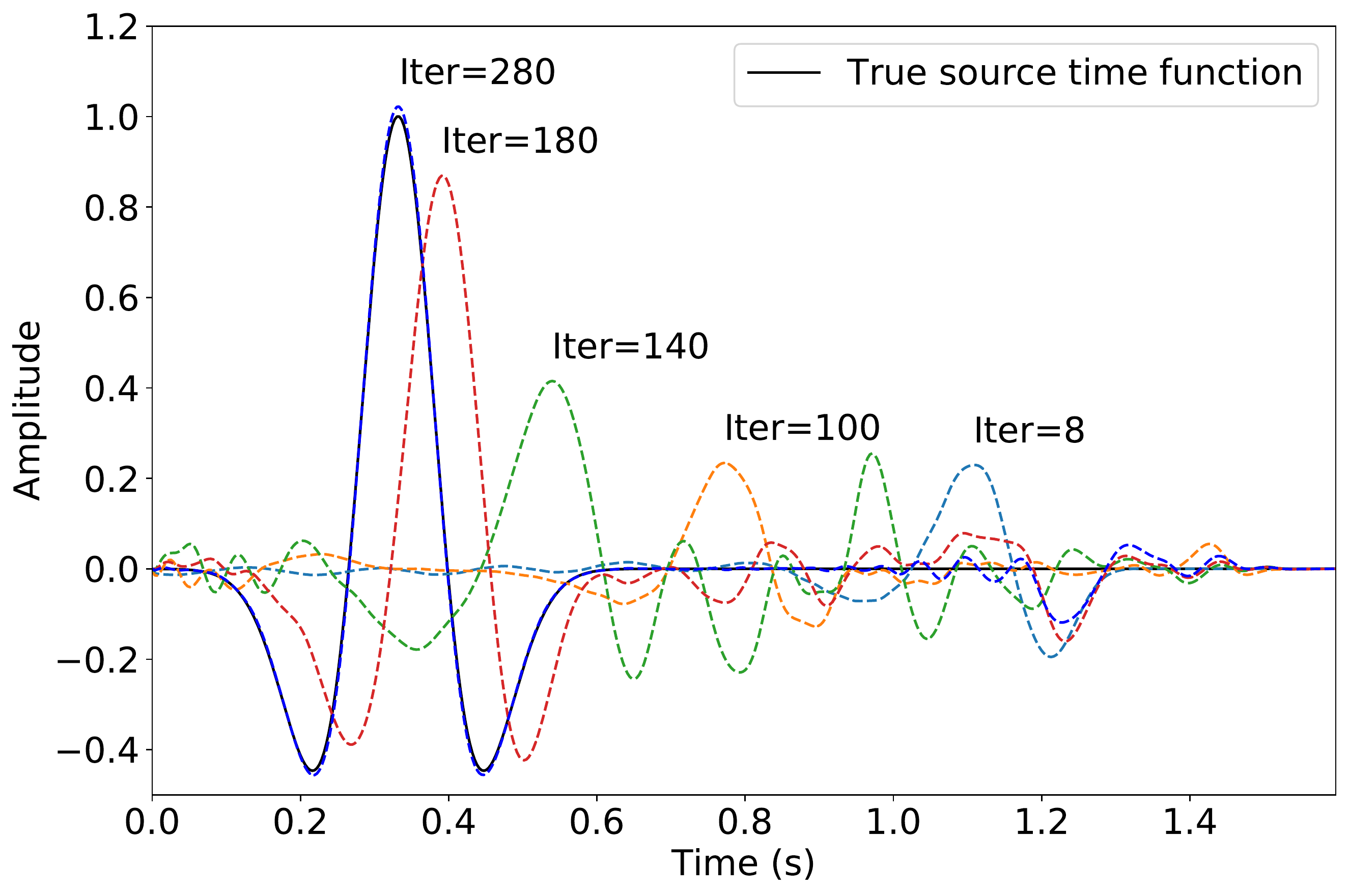}
		\caption{}
		\label{fig:source_func}
	\end{subfigure}	
	\caption{Inversion of earthquake location and source time function: (a) the velocity model and true source location; (b) the evolution of earthquake location represented by the black x; (c) the evolution of the source time function from a zero initial state.}
	\label{fig:source}
\end{figure*}

\subsection*{Earthquake Rupture Imaging}

The rupture process of large earthquakes has resolvable spatial and temporal extent. Imaging this rupture process from observed seismic data contributes to the understanding the complexity behind the evolution of earthquakes. The linearized kinematic inversion method using elastodynamic Green's functions ~\citep{kikuchi1982inversion, hartzell1983inversion, Beroza1988, beroza1991near, suzuki2011rupture, wald1990rupture, zhang2009spatio} and direct imaging methods, such as back-projection~\citep{ishii2005extent, lay2010teleseismic, xu2009rupture, kruger2005tracking, walker2005rupture, simons20112011, meng2012earthquake}, are the two most commonly used for imaging the earthquake rupture process. The adjoint-state method has also been tested for rupture process inversion \citep{kremers2011exploring, somala2018finite}. 

We consider a simplified 2D earthquake rupture case to show the potential applications of ADSeismic for imaging the earthquake rupture process. We mimic a simple rupture process with a group of sources activated from the left to right with different rise times and amplitudes (Fig.~\ref{fig:rupture_vp} and \ref{fig:rupture_true}). We consider two inversion targets: the entire rupture history, and the rupture time and amplitude. To estimate the rupture history, we choose the unknown parameter as the source time function ($f(t)$). To estimate the rupture time and amplitude, we choose the parameters of rupture time $t_0$ and amplitude $A_0$ by assuming that the shape of the source time function is known as a Gaussian function:
\begin{equation}
    f(t) = A_0 \exp \left( -\frac{(t-t_0)^2}{2 \sigma^2} \right)
\end{equation}
Imaging the entire rupture history contains many more parameters (the number of time steps $Nt$ for each candidate location) than when estimating only the rupture time and amplitude (two parameters ($A_0$ and $t_0$) for each candidate location), with the result that the former problem is less constrained for the same number of receivers.
To estimate the entire rupture history, the initial state is set to be zero slip for all locations (Fig.~\ref{fig:rupture_time_init}).
When estimating the rupture time and amplitude, the initial state is set to be a constant rupture time and amplitude.
The final inversion results are shown in Fig.~\ref{fig:rupture_inv} and \ref{fig:rupture_time_inv}. 
Note that we have not incorporated a dynamic rupture model to simulate the rupture propagation in this test; rather, AD provides an inversion method to back-propagate the gradients from the wave equation into the dynamic rupture equation to optimize the fault parameters based on seismic waves.

\begin{figure*}[htpb]
	\centering
	\begin{subfigure}{0.6\textwidth}
		\includegraphics[width=\textwidth]{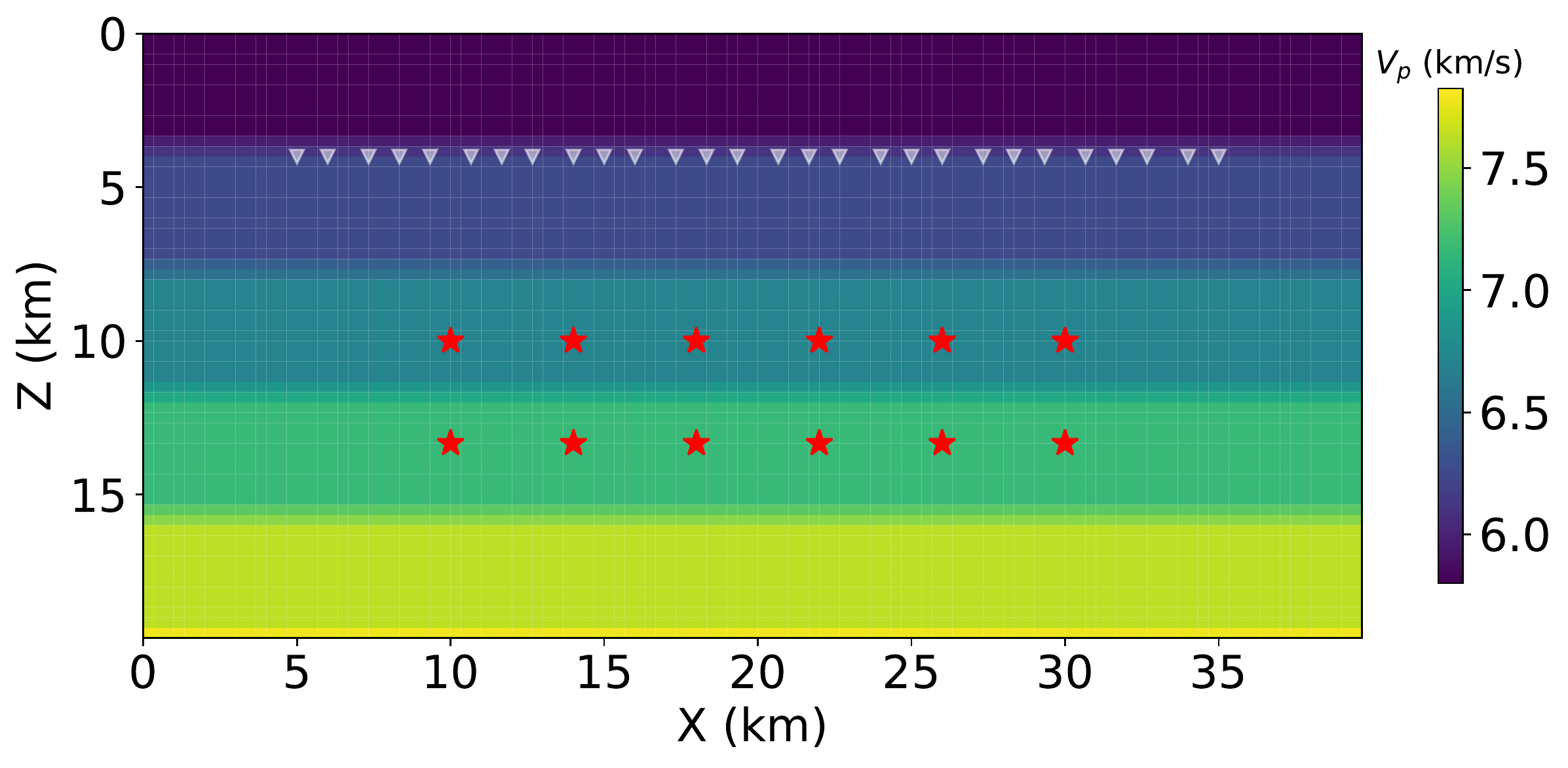}
		\caption{}
		\label{fig:rupture_vp}
	\end{subfigure}
	\begin{subfigure}{0.4\textwidth}
		\includegraphics[width=\textwidth]{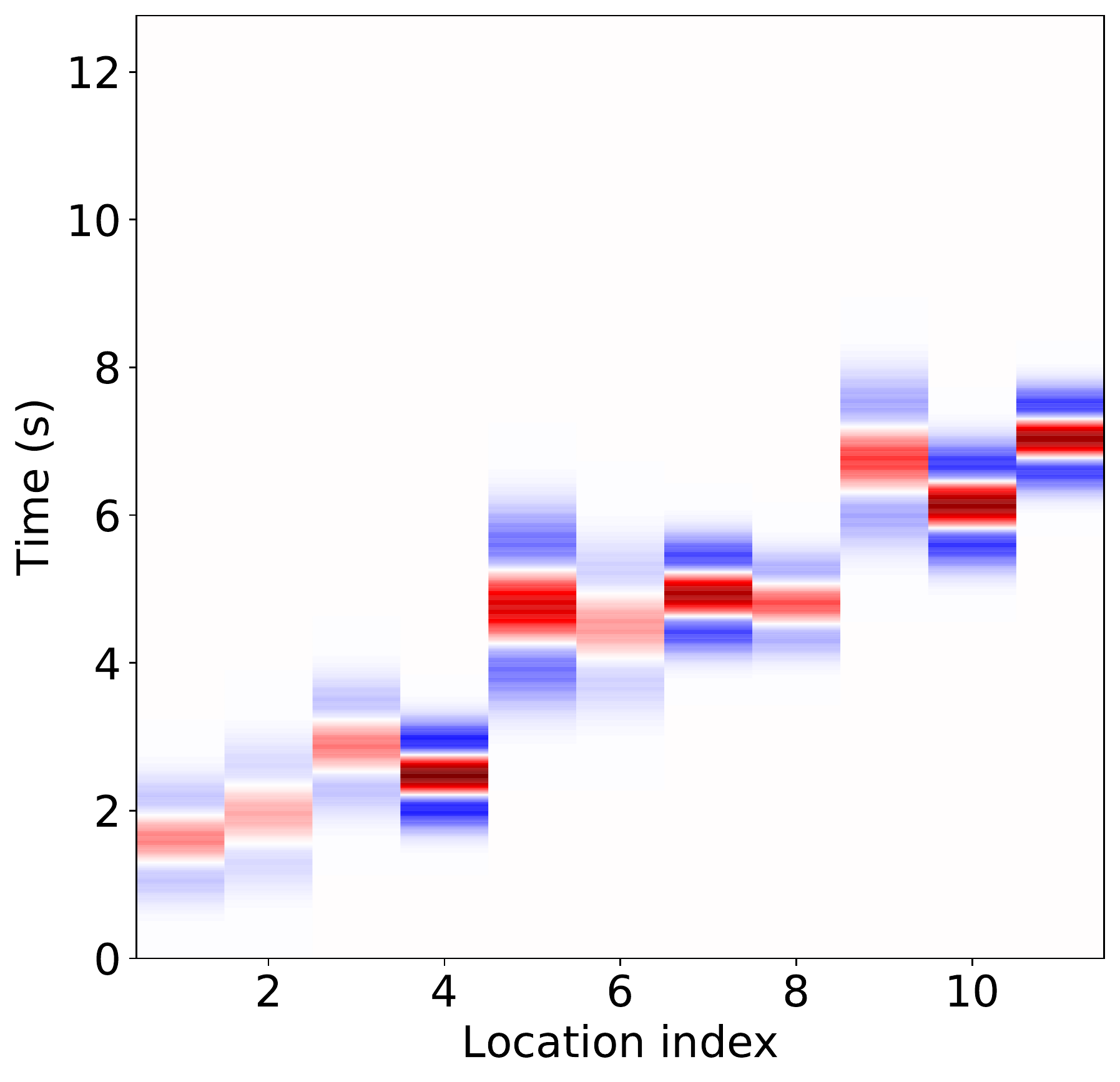}
		\caption{}
		\label{fig:rupture_true}
	\end{subfigure}
	\begin{subfigure}{0.4\textwidth}
		\includegraphics[width=\textwidth]{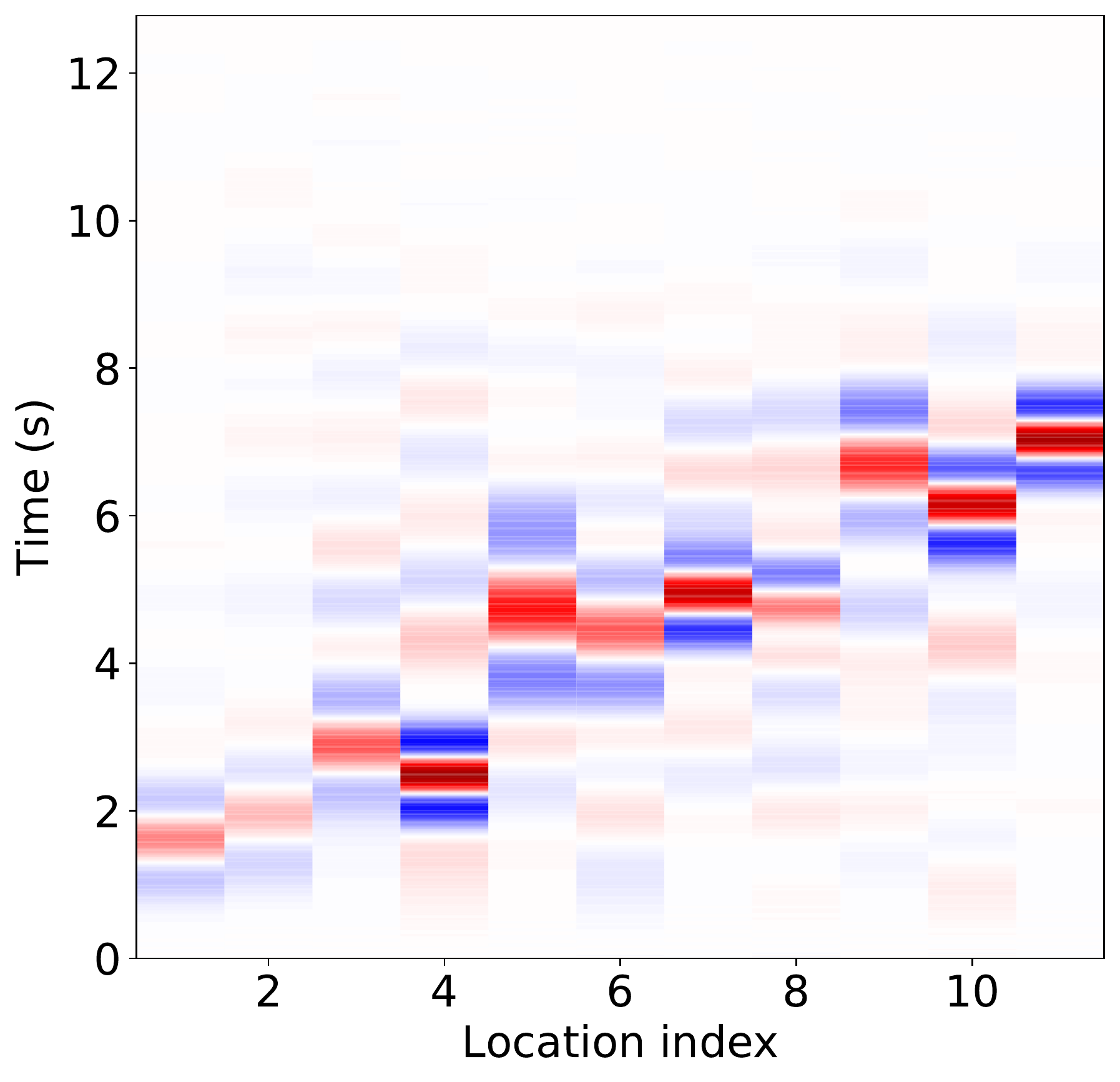}
		\caption{}
		\label{fig:rupture_inv}
	\end{subfigure}	
	\caption{Inversion of the whole rupture history: (a) the velocity model, receivers (white triangles), and simplified rupture locations (red starts); (b) true slip waveforms; (c) inverted slip waveforms. }
    \label{fig:rupture}
\end{figure*}

\begin{figure*}[htpb]
	\centering
	\begin{subfigure}{0.33\textwidth}
		\includegraphics[width=\textwidth]{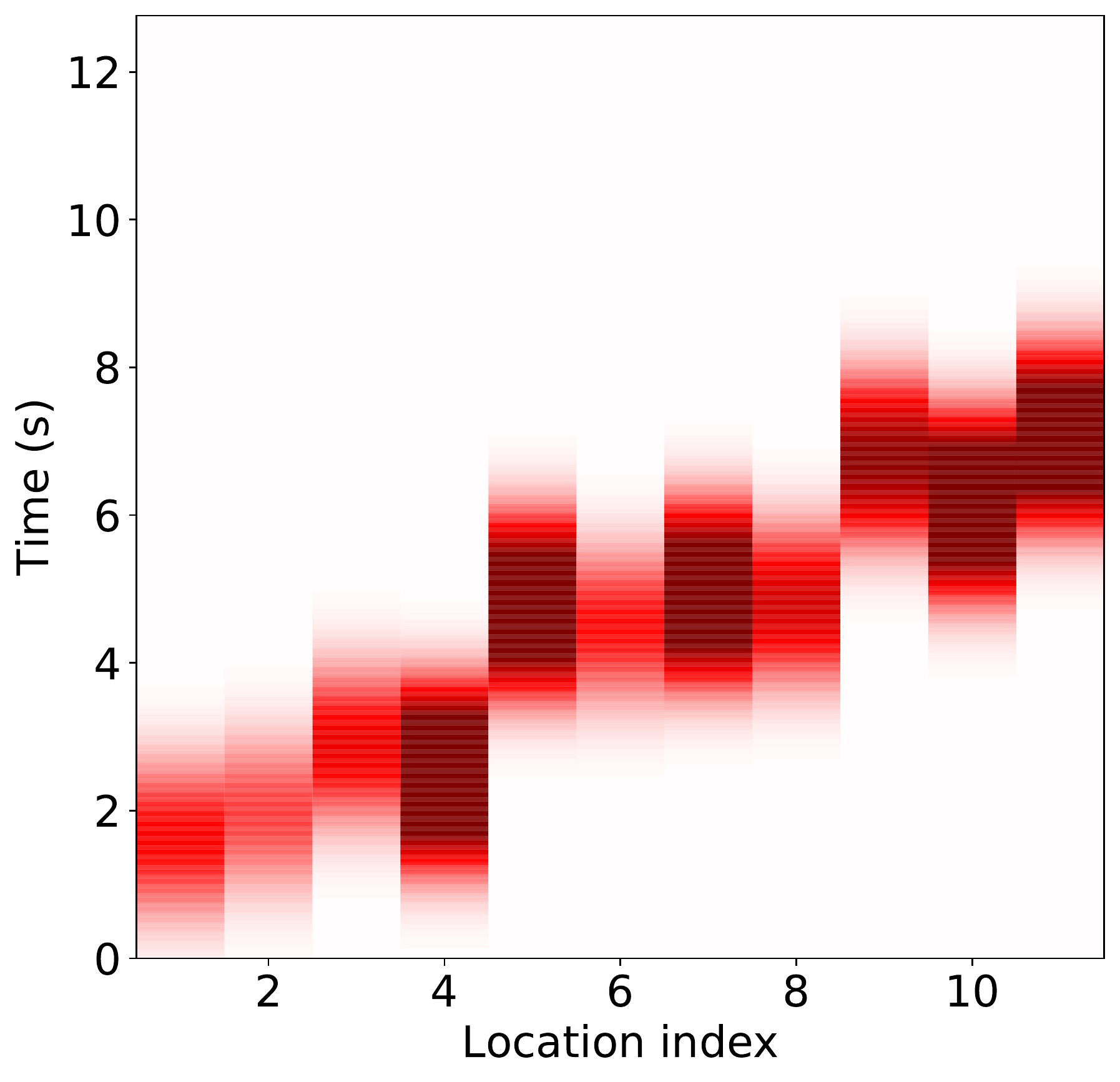}
		\caption{}
	\end{subfigure}~
	\begin{subfigure}{0.33\textwidth}
		\includegraphics[width=\textwidth]{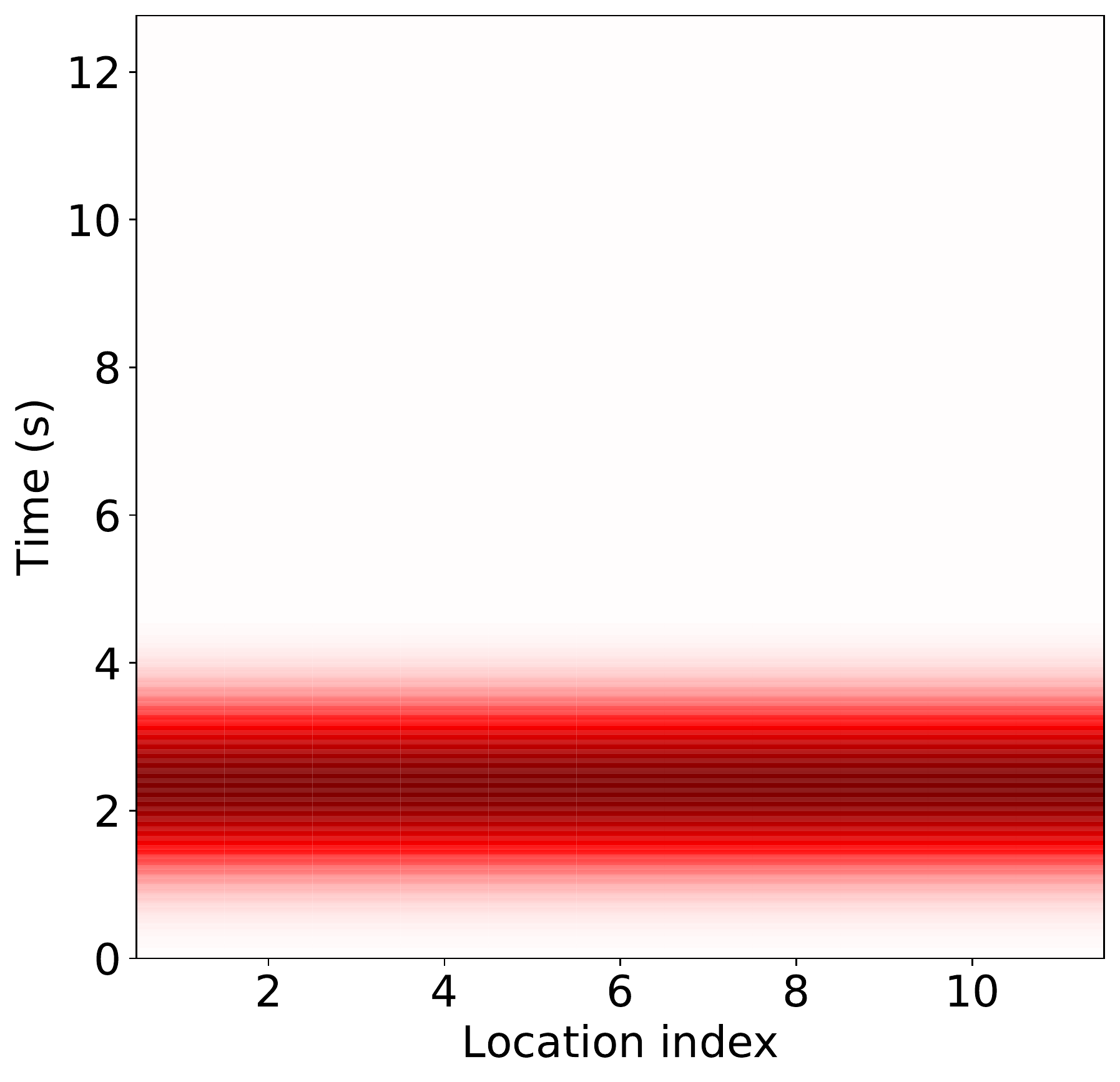}
		\caption{}
		\label{fig:rupture_time_init}
	\end{subfigure}~
	\begin{subfigure}{0.33\textwidth}
		\includegraphics[width=\textwidth]{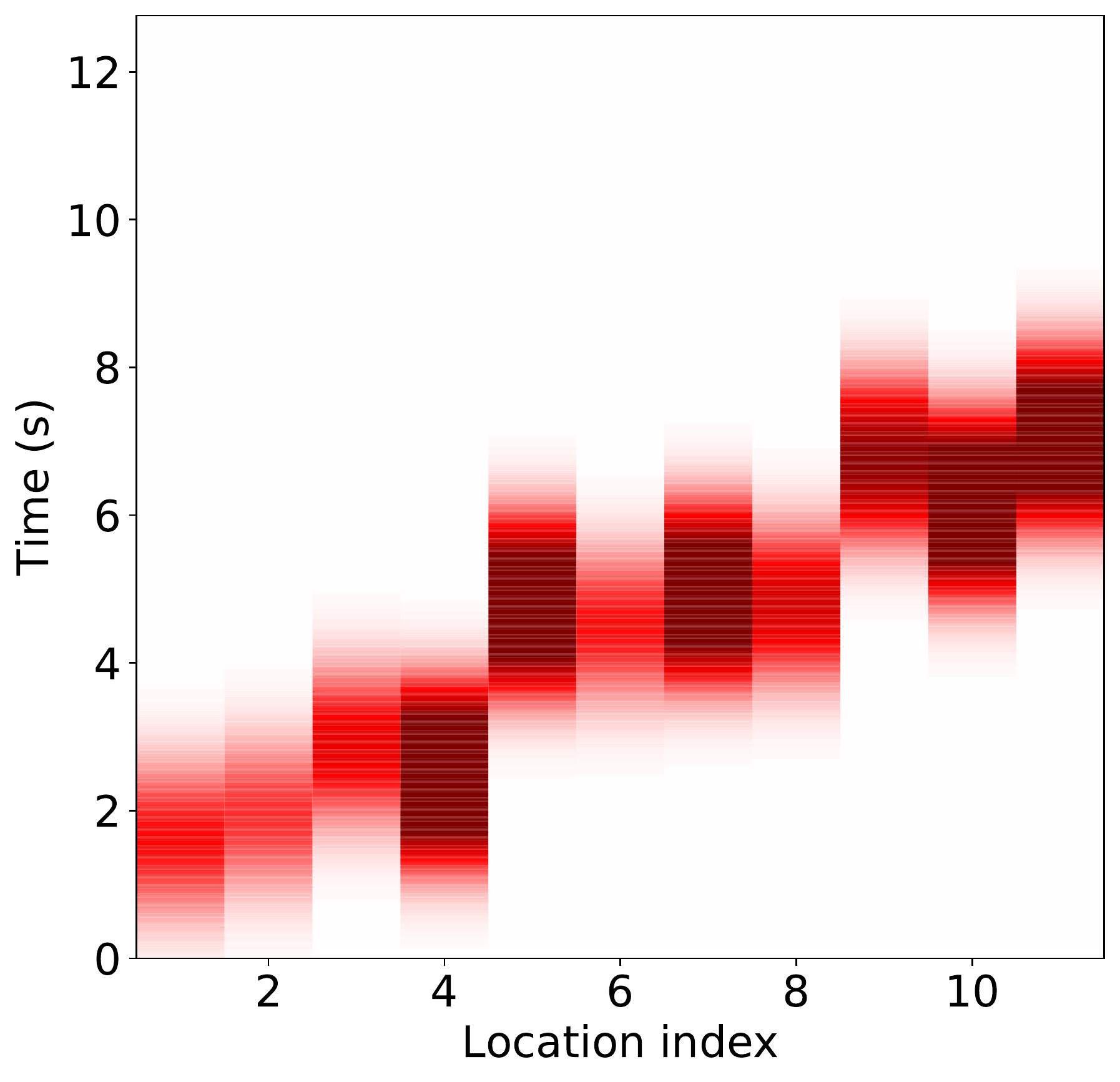}
		\caption{}
		\label{fig:rupture_time_inv}
	\end{subfigure}	
	\caption{Inversion of slip time and amplitude: (a) true earthquake slip with a shape of Gaussian function; (b) initial inversion state with a same slip time and amplitude; (b) inverted earthquake slip.}
	\label{fig:rupture_time}
\end{figure*}

\section*{Limitations}

Despite the many strengths of ADSeismic, it has three major limitations:

First, as with any estimation problem, it may suffer from ill-conditioning. The same applies for the often-encountered problem in seismology of cycle-skipping~\citep{virieux2009overview, hu2018retrieving}, which produces a local minimum when the predicted signal is shifted more than half a wavelength from the observation due to a poor initial model or lack of low frequency information. Neither AD nor adjoint-state methods can solve the ill-conditioning issue, which is intrinsic to the optimization problem. Nevertheless, many of the techniques for improving the conditioning of the optimization problem \citep{biondi2014simultaneous, ma2013wave, wu2014seismic, yang2018application} can be applied to our AD framework. 

Second, reverse-mode AD has demanding memory requirements, which is a noteworthy constraint when running large simulations on GPUs. Techniques such as check-pointing schemes~\citep{chen2016training} have been used to to reduce memory requirements. In ADSeismic, we partially alleviate this problem by using multi-GPUs, where the source functions are split onto multiple GPUs and simulations are executed concurrently. 

Third, the numerical schemes we consider in this work are all explicit. In some applications \citep{richardson2018seismic, liu2009practical, chu2012implicit}, implicit schemes are desirable for reasons such as stability, accuracy, and nonlinearity. For implicit schemes it is challenging to apply reverse-mode AD techniques since most AD frameworks only provide explicit differentiable operators. \citet{li2019time} introduce the intelligent automatic differentiation method that implements AD for implicit numerical schemes.  This approach could be used for augmenting ADSeismic for implicit schemes.

\section*{Conclusion}

We have demonstrated the connection between the automatic differentiation technique in deep learning and adjoint-state methods in seismic numerical simulations. Based on that correspondence we design a general seismic inversion framework, ADSeismic, based on the AD functionality from deep learning software. ADSeismic shows promising results on a series of seismic inversion problems and demonstrates dramatic acceleration on GPUs compared with CPUs. 
Since deep learning techniques and frameworks are continuously improving, ADSeismic allows for flexibly experimenting with new models, leverages specialized hardware designed for deep learning, and executes numerical simulations on heterogeneous computing platforms. This should facilitate general seismic inversion in a high performance computing environment. Furthermore, it opens a pathway for innovation in inverse modeling in geophysics by leveraging AD functionalities in a deep learning framework.


\clearpage
\bibliographystyle{apacite}
\bibliography{reference}

\end{document}